\documentclass[11pt]{article}  

\usepackage{xspace}
\usepackage{verbatim}
\usepackage[usenames,dvipsnames]{color}
\usepackage{graphicx}
\usepackage{epsfig}
\usepackage{tabularx}

\usepackage{amsmath}
\usepackage[showonlyrefs]{mathtools}
\usepackage{empheq}
\usepackage{amstext,amssymb,amsfonts}
\usepackage{fullpage}
\usepackage{nicefrac}
\usepackage[noend]{algorithmic}
\usepackage{algorithm}
\usepackage{bm}
\usepackage[normalem]{ulem}

\usepackage{nameref}
\usepackage[capitalize]{cleveref}
\usepackage{url}
\usepackage{textcomp,setspace}
\usepackage{amsthm}

{\unskip\nobreak\hskip 2em plus 1fil\nobreak$\Box$
\parfillskip=0pt%
\endtrivlist}

\usepackage{thmtools}
\declaretheorem[within=section]{theorem}

\declaretheorem[sibling=theorem]{lemma}

\declaretheorem[sibling=theorem]{proposition}

\newcounter{termcounter}
\renewcommand{\thetermcounter}{\Alph{termcounter}}
\crefname{term}{term}{terms}
\creflabelformat{term}{\textup{(#2#1#3)}}

\makeatletter
\def\term{\@ifnextchar[\term@optarg\term@noarg}%]
\def\term@optarg[#1]#2{%
  \textup{(#1)}%
  \def\@currentlabel{#1}%
  \def\cref@currentlabel{[][2147483647][]#1}%
  \cref@label[term]{#2}}
\def\term@noarg#1{%
  \refstepcounter{termcounter}%
  \textup{(\thetermcounter)}%
  \cref@label[term]{#1}}
\makeatother
\newcommand{\nfrac}{\nicefrac}
\newcommand{\half}{\nfrac12}

\newcommand{\ignore}[1]{}

\newcommand{\argmin}{\mathop{\mathrm{argmin}}}

\newcommand{\poly}{\mathrm{poly}}

\newcommand{\set}[1]{\{#1\}}

\renewcommand{\vec}[1]{{\mathbf{#1}}}

\definecolor{DSred}{rgb}{1,0,0}

\renewcommand{\leq}{\leqslant}
\renewcommand{\geq}{\geqslant}
\renewcommand{\ge}{\geqslant}
\renewcommand{\le}{\leqslant}
\renewcommand{\epsilon}{\varepsilon}

\newcommand{\R}{\mathbb{R}}

\newcommand{\Z}{\mathbb{Z}}

\begin{document}
\title{On the Convergence of the Hegselmann-Krause System}
\author{Arnab Bhattacharyya\thanks{DIMACS, Rutgers
    University. Research done while at Princeton University and Center
    for Computational Intractability, supported by NSF Grants CCF-0832797, 0830673, and 0528414. \texttt{arnabb@dimacs.rutgers.edu}.}
\and
Mark Braverman\thanks{Department of Computer Science, Princeton University, \texttt{mbraverm@cs.princeton.edu}. Research supported in part by an Alfred P. Sloan Fellowship, an NSF CAREER award, and a Turing Centenary Fellowship.}
\and
Bernard Chazelle\thanks{Princeton University and Coll\`ege de France, \texttt{chazelle@cs.princeton.edu}. This work was supported in part by NSF grants CCF-0832797, CCF-0963825, and CCF-1016250.}
\and
Huy L. Nguy$\tilde{\hat{\mbox{e}}}$n\thanks{Department of Computer Science, Princeton University, \texttt{hlnguyen@princeton.edu}. Supported in part by NSF CCF-0832797 and a Gordon Wu fellowship.}}
\maketitle

\begin{abstract}
We study convergence of the following discrete-time non-linear dynamical system:
$n$ agents are located in $\R^d$ and at every time step,
each moves synchronously to the average location of all agents within
a unit distance of it. This popularly studied system was introduced by
Krause to model the dynamics of opinion formation and is often
referred to as the {\em Hegselmann-Krause model}. We prove the first
polynomial time bound for the convergence of this system in arbitrary
dimensions. This improves on the bound of $n^{O(n)}$ resulting from 
a more general theorem of Chazelle \cite{Chazelle11}. Also, we show a
quadratic lower bound and improve the upper bound for one-dimensional
systems to $O(n^3)$. 
\end{abstract}

\section{Introduction}
Lately, there has been a surge of attention given to network-based dynamical systems,
in which agents interact according to local rules via a dynamic communication
graph~\cite{EK10}. Much of the previous work has focused on the
exogenous case, where the communication topology is decoupled from
the evolution of the system.  We refer interested readers
to~\cite{BHT09} for a good overview of research in this area.
In the more common, endogenous version, the communication graph changes over time according to the
current states. The feedback loop between dynamics and topology
creates considerable difficulties, and efforts have been underway 
to build an algorithmic calculus within the broad  
framework of {\em influence systems}~\cite{chazelle-infl}.
This work investigates the complexity
of {\em Hegselmann-Krause systems} (abbreviated as HK system from now on), 
a popular model of opinion dynamics that has proven highly 
influential over the years~\cite{Axelrod97, Fortunato05, HK02, Krause00, Lorenz07}
and stands as the archetype of a diffusive influence system~\cite{Ccacm12}.
In the $d$-dimensional version of the model, each agent has an opinion
represented as a point in $\R^d$. Two agents are neighbors if they
are within unit distance from each other. At every time step, each
agent moves synchronously to the mass center of its neighbors.

HK systems are known to converge~\cite{HendrickxB06, Lorenz05, Moreau05},
meaning that they eventually come to a full stop.
The convergence time has been bounded by $n^{O(n)}$ time 
and conjectured to be polynomial~\cite{Chazelle11}.
It was shown to be $O(n^5)$ in the case $d=1$~\cite{MBCF07}.
The contribution of this work is threefold: first, we
prove that the convergence time is indeed polynomial in $n$,
regardless of the dimension; second, we lower the one-dimensional bound
from $O(n^5)$ to $O(n^3)$; third, we establish a quadratic lower bound,
which improves the known bound of $\Omega(n)$~\cite{MBCF07}.
We also consider noisy variants of the model.

The bidirectionality of the system plays a crucial role in the proof,
as it should. Indeed, it is known from~\cite{chazelle-infl}
that allowing different radii and averaging weights for the agents
can prevent the convergence of the communication graph.
Our proofs are an elementary mix of geometric and algebraic techniques.
Much of the current technology for HK systems centers around
products of stochastic matrices. This work injects a geometric
perspective that, we believe, will be necessary for further progress
on the more difficult directional case.

\section{Preliminaries}

We formally define the discrete-time {\em HK system in dimension $d$}  as
follows. There are $n$ agents. For every $t\in \Z^{\geq
  0}$ and for every $i \in [n]$, the {\em position} of agent $i$ at
time $t$ is $\vec{x}_i(t) \in \R^d$. The positions at $t=0$ are given
and, thereafter, are updated synchronously according to the following rule:
\begin{equation}
\vec{x}_i(t+1) = \frac{\sum\limits_{j: \|\vec{x}_i(t) - \vec{x}_j(t)\| \leq
    1} \vec{x}_j(t)} {\sum\limits_{j: \|\vec{x}_i(t) - \vec{x}_j(t)\| \leq 1} 1}
\end{equation}
Here, $\|\cdot\|$ denotes the Euclidean norm.
We say that agents $i$ and $j$ are {\em neighbors} at time $t$ if
$\|\vec{x}_i(t) - \vec{x}_j(t)\| \leq 1$, and we denote the {\em
  neighborhood} at time $t$ of agent $i$ by $\mathcal{N}_i(t) = \{j :
i\text{ and }j\text{ neighbors at time }t\}$. Note that if $i \in
\mathcal{N}_j(t)$, then $j \in \mathcal{N}_i(t)$. Also for a given
$\vec{x} \in \R^d$, we define the {\em weight of $\vec{x}$ at time
$t$} to be $w_t(\vec{x}) = |\set{i: \vec{x}_i(t) = \vec{x}}|$. Note
that at any given $t$, $0\leq w_t(\vec{x})\leq n$ for all $\vec{x}
\in \R^d$ and that the sum of all weights equals $n$. Also, let the
{\em weight of agent $i$ at time $t$} denote $w_t(\vec{x}_i(t))$. 
The weight of an agent is monotonically non-decreasing with time.
The system is said to have {\em converged} at time $t$ if $\vec{x}_i(t+1) =
\vec{x}_i(t)$ for all $i \in [n]$.

In the case of one dimension, one can observe that order is preserved:
\begin{proposition}[Lemma 2 in \cite{Krause00}]\label{prop:order}
The $HK$ system in dimension $1$ preserves order of positions. That
is, for any $i, j \in [n]$, if $x_i(0) \leq x_j(0)$, then $x_i(t) \leq x_j(t)$ for all $t > 0$.
\qed
\end{proposition}
For this reason, in one dimension, we can number the agents from $1$
to $n$ such that $x_1(t) \leq x_2(t) \leq \cdots \leq x_n(t)$ for all $t
\geq 0$. The one-dimensional system also has the following
decomposability property:
\begin{proposition}[Proposition 2 in \cite{BHT09}]\label{prop:dec}
Suppose $d=1$.
If at time $t$, and for some $i\in [n]$, $|x_{i+1}(t) - x_i(t)| > 1$,
then for all $t'>t$ also, $|x_{i+1}(t') - x_i(t')| > 1$. So, the
system can be decomposed into two subsystems, one consisting of agents
$\set{1,\dots,i}$ and the other of agents $\set{i+1,\dots,n}$, each
evolving independently after time $t$.\qed
\end{proposition}
Therefore, if at time $t$, an agent has no neighbor strictly to its
left and no neighbor strictly to its right, it never moves
subsequently, and we say that the agent has {\em frozen} at time $t$.

This decomposability property is not true in higher dimensions. For instance, consider the
two-dimensional HK system with agents at positions $(0,0)$, $(1,0)$ and
$(\half,1)$; the third agent does not neighbor any other
agent at $t=0$, but this is not so at $t=1$.

\section{Convergence in One Dimension}

\begin{theorem}\label{thm:1dimub}
The HK system in $1$ dimension converges within $O(n^3)$ time steps.
\end{theorem}
\begin{proof}
Recall that we number the agents such that $x_i(t) \leq x_{i+1}(t)$
for all $t$ and all $i \in [1,n-1]$. 
Suppose the system has not already converged at time $t$, and let
$\ell(t)$ denote the leftmost non-frozen agent, i.e., the least $\ell \geq 1$ such that 
$\set{x_i(t): i \in \mathcal{N}_\ell(t)} \neq \set{x_\ell(t)}$. Note
that at time $t$, agent $\ell(t)$ must have at least one neighbor strictly to its
right and no neighbor strictly to its left (because any neighbor
strictly to the left would violate minimality of $\ell$). 

The following lemma is the heart of our proof.
\begin{lemma}\label{lem:leftmove}
For every $t \geq 0$, by time $t+2$, agent $\ell(t)$ increases
in weight, or gets frozen, or moves to the right by at least
$\frac{1}{2n^2}$. 
\end{lemma}
\begin{proof}
Fix $t$, let $\ell = \ell(t)$, and let $r = \min \set{j: x_j(t) >
  x_\ell(t)} $ be the leftmost agent that is strictly right of agent
$\ell$. Agent $r$ is a neighbor of $\ell$ because $\ell$ has at least
one neighbor strictly to its right. If $\mathcal{N}_\ell(t) =
\mathcal{N}_r(t)$, agent $\ell$ at time $t+1$ moves to the same
location as agent $r$ does and increases its weight. So, suppose otherwise. Since $\ell$ has no
neighbor strictly to its left, $r$ also does not have any neighbor
strictly left of $\ell$. Hence, there must exist $s \in
\mathcal{N}_r(t)\setminus \mathcal{N}_\ell(t)$ strictly to the right
of $r$. 

We now show $x_r(t+1) \geq x_\ell(t) +\frac{1}{n}$. Observe that
$x_s(t) - x_\ell(t) > 1$. So, if $\delta =  x_r(t) - x_\ell(t)$ and $k
= |\mathcal{N}_r(t)|$,  then agent $r$ moves to the left by at most:
\begin{equation*}
\frac{\delta\cdot (k-1) - (1-\delta)\cdot 1}{k} = \delta - \frac{1}{k}
\leq \delta -\frac{1}{n}
\end{equation*}

If $x_\ell(t+1) \geq x_\ell(t)+\frac{1}{2n}$, we are already
done. Otherwise, $x_r(t+1)-x_\ell(t+1) \geq \frac{1}{2n}$. By
\cref{prop:dec}, agent $\ell$ still has no neighbor strictly to its left at
time $t+1$. If $x_r(t+1)-x_\ell(t+1) > 1$, then agent $\ell$ gets
frozen at time $t+1$. Otherwise, $\frac{1}{2n} \leq
x_r(t+1)-x_\ell(t+1) \leq 1$, and so:
$$x_\ell(t+2) - x_\ell(t+1) \geq \frac{1}{2n|\mathcal{N}_\ell(t+1)|}
\geq \frac{1}{2n^2}$$
\end{proof}

We can assume $x_n(0)-x_1(0) \leq n$ without loss of generality,
because otherwise, by \cref{prop:dec}, the system can be decomposed
into independently evolving subsystems. So, $x_n(t) - x_1(t) \leq n$
for all $t$. Now, apply
\cref{lem:leftmove} at $t=0,2,4,\dots$ as long as the system has not
converged and $\ell(t)$ exists. $\ell(t)$ can increase
in weight only at most $n$ times. Also, because $\ell(t)$ is
non-decreasing with $t$, $x_{\ell(t+2)}(t+2)\geq x_{\ell(t)}(t+2)$
and, hence, the third case in \cref{lem:leftmove} can occur only at
most $2n^3$ times.  Thus, after $t>2(n+2n^3)$, \cref{lem:leftmove}
cannot be applied, and the system must have converged.
\end{proof}

\subsection{Noisy neighborhoods}

In this subsection, we study an extension to a noisy version of the
HK model. We consider the following system HK$_\eta$ where $\eta \in 
(0,1)$ is a parameter. There are $n$ agents, and each agent has an
associated {\em left-neighborhood parameter} $\eta_i$ that is in the
interval $(0,\eta)$. The positions at $t=0$ are given, and then, the
positions are updated according to the following rule:
\begin{equation}
x_i(t+1) = \frac{\sum\limits_{j: -1+\eta_i \leq x_j(t) - x_i(t) \leq 1}x_j(t)}{\sum\limits_{j:  -1+\eta_i \leq x_j(t) - x_i(t) \leq 1}1}
\end{equation}
One can interpret the $\eta_i$'s as noise acting on the left side of
each agent's neighborhood. We can prove the following extension of
\cref{thm:1dimub} to HK$_\eta$.
\begin{theorem}\label{thm:noisyub}
For fixed $\eta \in (0,1)$, the HK$_\eta$ system  converges within
$O(n^3)$ time steps. 
\end{theorem}
\begin{proof}
The proof proceeds similarly to that of \cref{thm:1dimub}. 
We say that agent $j$ is a neighbor of agent $i$ at time $t$ if  $-1 +
\eta_i \leq x_j(t) - x_i(t) \leq 1$. 
Conceptually, the main complicating issue is that being neighbors is no longer a
symmetric relation. Also, \cref{prop:order} is no longer true as order
may not be preserved by the dynamics.

Let $\ell(t)$ denote the leftmost agent at time $t$ that neighbors at least one
agent positioned strictly to its right. Notice that this implies there
is no $k$ such that $0<x_{\ell(t)}-x_k \leq 1$, because agent $k$ would
neighbor $\ell(t)$ and hence violate the minimality of $\ell(t)$.
 We prove the following analog of \cref{lem:leftmove}.
\begin{lemma}\label{lem:leftmoveext}
For every $t \geq 0$, by time $t+2$, agent $\ell(t)$ increases
in weight, or gets frozen, or moves to the right by at least
$\frac{1}{2n^2}$. 
\end{lemma}
\begin{proof}
Let $\ell = \ell(t)$ and let $r$ be the leftmost neighbor of $\ell$. 
If $x_r(t) - x_\ell(t) > 1-\eta_r$, then $x_\ell(t+1) > \nfrac{(1-\eta_r)}{n}
\geq \nfrac{(1-\eta)}{n}$, and we are already done. So, assume
otherwise. 

Then, agent $\ell$ is a neighbor of agent $r$. Since $\ell$ has no
neighbor strictly to its left within a distance of $1$, $r$ also has no
neighbor strictly to the left of $\ell$. Now, we can use the same
argument as in the proof of \cref{lem:leftmove} to argue that either
agent $\ell$ freezes at time $t+1$ or at time $t+2$, it moves to the
right by at least $\frac{1}{2n^2}$.
\end{proof}
The rest of the proof can be finished in exactly the same way as the
proof of \cref{thm:1dimub}.
\end{proof}

By symmetry, \cref{thm:noisyub} also applies when the right side of
each agent's neighborhood is perturbed and the left side is fixed. The
case when both sides are perturbed remains open; convergence for such
heterogeneous HK systems is conjectured \cite{MB11}.

\section{Convergence in Higher Dimensions}

In this section, we consider the HK system in $d$ dimensions with $n$ agents
and $d \ge 2$. We show that the convergence time is polynomial in both
$n$ and $d$. 

The proof follows from a sequence of lemmas. The first lemma asserts
that there is some vector $\vec{a}$ such that taking the 
projection of agents on $\vec{a}$ does not bring two non-neighbors
too close together.
\begin{lemma}\label{lem:gooddir}
  For every $t$, there exists a unit vector $\vec{a}$ such that for any two
  agents $i, j$, we have $\left|\vec{a}\cdot
  \frac{\vec{x}_i(t)-\vec{x}_j(t)}{\|\vec{x}_i(t)-\vec{x}_j(t)\|}\right| = \Omega(n^{-2}d^{-1})$. 
\end{lemma}

\begin{proof}
  Let $\vec{d}_{i,j}=\frac{\vec{x}_i(t)-\vec{x}_j(t)}{\|\vec{x}_i(t)-\vec{x}_j(t)\|}$. We can view $\vec{d}_{i,j}$ as a point
  on $S^{d-1}$, the unit ball in $\R^d$. The set of points on $S^{d-1}$ with dot product with
  $\vec{d}_{i,j}$ of absolute value at most $O(n^{-2}d^{-1})$ has area
  $\frac{\pi^{(d-2)/2}}{n^2 d\cdot \Gamma((d-2)/2)}$. The area of
  $S^{d-1}$ is $\frac{2\pi^{d/2}}{\Gamma(d/2)}$. There are $n\choose
  2$ pairs $i,j$, and so, by a volume argument, there exists a point $\vec{a}\in
  S^{d-1}$ such that for all $i$ and $j$, $|\vec{d}_{i,j}\cdot \vec{a}| = \Omega(n^{-2}d^{-1})$. 
\end{proof}

The next lemma analyzes the one-dimensional system formed by
projecting onto $\vec{a}$ and shows a lower bound on the total
movement in each step.
\begin{lemma}\label{lem:movefar}
  Assume that at time $t$, the system has not converged. One of the following cases happens.
  \begin{itemize}
  \item Two agents move to the same location at time $t+1$.
  \item Some agent moves by at least $\Omega(n^{-4}d^{-1})$
  \end{itemize}
\end{lemma}

\begin{proof}
Let $\vec{a}$ be a unit vector satisfying \cref{lem:gooddir} and $j$
be the agent with minimum $\vec{x}_j(t)\cdot \vec{a}$ among agents with neighbor
at time $t$. Consider two cases. 

\paragraph{Case 1.} ${j}$ has a neighbor $i$ such that
$(\vec{x}_i(t)-\vec{x}_j(t))\cdot \vec{a} = \Omega(n^{-3}d^{-1})$. Because
$(\vec{x}_k(t)-\vec{x}_j(t))\cdot \vec{a} \ge 0~\forall k$, $j$ must move by
$\Omega(n^{-4}d^{-1})$. 

\paragraph{Case 2.} All neighbors $i$ of $j$ satisfy
$(\vec{x}_i(t)-\vec{x}_j(t))\cdot \vec{a} = O(n^{-3}d^{-1})$. If none of them
has any neighbor that is not $j$ or $j$'s neighbors, then they all
move to the same location at time $t+1$. Otherwise, some neighbor $i$
of $j$ has a neighbor $k$ that is not a neighbor of $j$. Then
$(\vec{x}_k(t)-\vec{x}_j(t))\cdot \vec{a} =\Omega(n^{-2}d^{-1})$. We have  
$$\vec{x}_i(t+1)\cdot \vec{a} \ge \vec{x}_i(t) \cdot \vec{a} +
\frac{(\vec{x}_k(t) -\vec{x}_j(t))\cdot \vec{a} +
  (n-1)(\vec{x}_j(t)-\vec{x}_i(t))\cdot \vec{a}}{n} \ge
\vec{x}_i(t)+\Omega(n^{-3}d^{-1})$$ 
Thus $i$ moves by at least $\Omega(n^{-3}d^{-1})$.
\end{proof}

We now use the following special case of a result from~\cite{RMF08}
to show the existence of a potential function which is strictly
decreasing as long as some agent moves. The proof is included for
completeness. 
\begin{theorem}[\cite{RMF08}, Theorem 2]\label{thm:potential}
  Let $f_{i, j}^t : \R^d \times \R^d \rightarrow \R$ be defined as
  follows: $f_{i,j}^t(\vec{z},\vec{z'}) = \|\vec{z}-\vec{z'}\|^2$ if
  $\|\vec{x}_i(t)-\vec{x}_j(t)\| < 1$ and $f_{i,j}^t(\vec{z},\vec{z'}) =
  1$ otherwise.
  Let $V(t) = \sum_{i, j \in [n]} f_{i,j}^t(\vec{x}_i(t),
  \vec{x}_j(t))$. Then $V$ is non-increasing along the trajectory of
  the system. In fact:
\begin{align*}
V(t) - V(t+1) &\ge \sum_{i, j: f_{i,j}^t \not\equiv 1} f_{i,j}^t (\vec{x}_i(t+1)-\vec{x}_i(t),\vec{x}_j(t)-\vec{x}_j(t+1))\\
&\ge \sum_{i} 4\|\vec{x}_i(t+1)-\vec{x}_i(t)\|^2
\end{align*}
\end{theorem}
\begin{proof}
First, notice that $f_{i,j}^{t+1}(\vec{x}_i(t+1), \vec{x}_j(t+1)) \le f_{i,j}^t (\vec{x}_i(t+1), \vec{x}_j(t+1))~\forall i,j$. Next, we will show

$$\sum_{i,j} f_{i,j}^t (\vec{x}_i(t), \vec{x}_j(t)) - f_{i,j}^t (\vec{x}_i(t+1), \vec{x}_j(t+1))= \sum_{i, j \in [n]: f_{i,j}^t \not\equiv 1} f_{i,j}^t (\vec{x}_i(t+1)-\vec{x}_i(t),\vec{x}_j(t)-\vec{x}_j(t+1))$$

Let $F:(\R^d)^n \times (\R^d)^n \rightarrow \R$ be defined as $F(\vec{x},\vec{y})
= \sum_{i,j:f_{i,j}\not\equiv 1} f_{i,j}^t (\vec{x}_i, \vec{y}_j)$. Then
$\vec{x}(t+1) = \argmin_\vec{y} F(\vec{x}(t), \vec{y})$. We can also view $F$ as 2 matrices
$A,B$ and $F(\vec{x},\vec{y}) = \vec{x}^T A \vec{x} + \vec{y}^T A \vec{y} -2\vec{x}^T B \vec{y}$ where $A,B$ are
symmetric and $A$ is positive semidefinite. We have $\vec{x}(t+1) = A^{-1}B\vec{x}(t)$. 
\begin{align*}
F(\vec{x}, \vec{x}) - F(A^{-1}B\vec{x}, A^{-1}B\vec{x}) &= 2\vec{x}^T(B A^{-1}(A-B)A^{-1}B + A - B)\vec{x}\\
&= 2\vec{x}^T(B A^{-1}B - BA^{-1}BA^{-1}B + A - B)\vec{x}\\
&= 2\vec{x}^T(I-B A^{-1})(A+B)(I-A^{-1}B)\vec{x}\\
&= F((I-A^{-1}B)\vec{x}, -(I-A^{-1}B)\vec{x})
\end{align*}
\end{proof}

\begin{theorem}
The system converges in $\poly(n,d)$ time.
\end{theorem}
\begin{proof}
  There are at most $n$ time steps where two agents move to the same
  place. In all other time steps, by \cref{lem:movefar}, some
  agent moves by at least $\Omega(n^{-4}d^{-1})$. We have $V(0) \le
  n^2$ and $V(t)\ge 0~\forall t$. Therefore, by
  \cref{thm:potential}, the number of time steps before the
  system converges is $O(n^{10}d^2)$. 
\end{proof}

\section{Lower bound}

In this section, we give an instance of an HK system 
that requires $\Omega(n^2)$ steps to converge. Our example is in two
dimensions, and we know of no example that takes longer time even
when the number of dimensions is large. The previous best lower bound on
the convergence time \cite{MBCF07} was $\Omega(n)$, which is achieved
by $n$ points on a line with unit distance spacing between each pair
of consecutive points. This is still the worst example we know of in
one dimension.

Our lower bound is achieved by the following system. There are $n$ agents in the
system and initially, they are located at vertices of a regular
$n$-gon with side length $l_1 = 1$. Label the agents clockwise around
the polygon from $0$ to $n-1$. Let $O$ be the center of the polygon
and $A_i$ be the location of agent $i$. For notational convenience,
all index computation is done modulo $n$ (so if $i=0$ then $A_{i-1}$
is $A_{n-1}$). 
\begin{theorem}
The system described above requires at least $\Omega(n^2)$ steps to converge.
\end{theorem}
\begin{proof}
We will prove by induction that for all $t \le n^2/28$, $A_i$'s are vertices of a regular $n$-gon with side length $l_t\ge(1-\frac{14}{n^2})^t$.

By the initial state of the system, the invariant holds for $t=1$. Assume that it holds before step $t=k$, we will show it holds before step $t=k+1 \le n^2/28$.

We analyze the movement of agent $i$ from $A_i$ to $A_i'$ in one step
(see \cref{fig:2dlb}). Note that $\angle A_i O A_{i+1} = \frac{2\pi}{n}$. For any $j\not\in\{i-1, i, i+1\}$, we have 
\begin{figure}
\begin{center}
\includegraphics{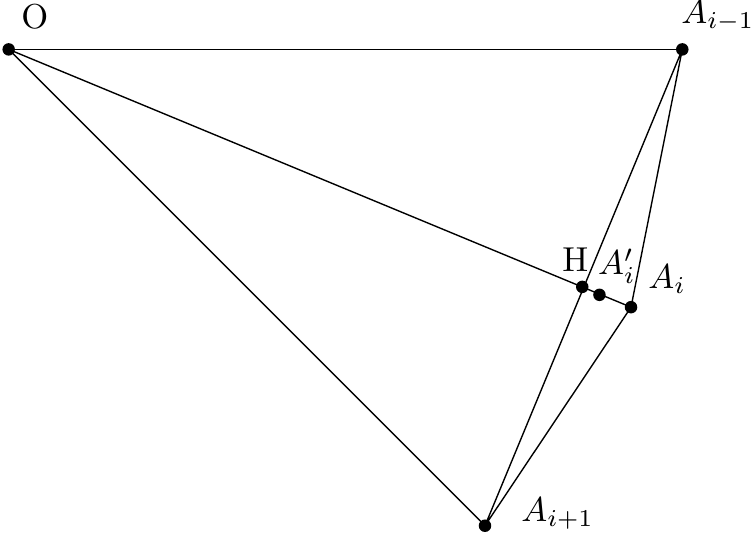}  
\end{center}
\caption{In one time step, agent $i$ moves from $A_i$ to $A_i'$, the
  centroid of $A_i, A_{i-1}, A_{i+1}$}
\label{fig:2dlb}
\end{figure}
\begin{align*}
A_iA_j &= \frac{OA_i\sin(\angle OA_iA_j)}{\sin(\angle A_i O A_j)}\\
&\ge \frac{OA_i\sin(\frac{4\pi}{n})}{\sin(\frac{\pi(n-4)}{2n})}\\
&= \frac{l_t\sin(\frac{\pi(n-2)}{2n})\sin(\frac{4\pi}{n})}{\sin(\frac{2\pi}{n})\sin(\frac{\pi(n-4)}{2n})}\\
&\ge 2 l_k\left(1-\frac{\pi^2}{2n^2}\right) \left(1-\frac{4\pi^2}{n^2}
  \right)\\
&> 1
\end{align*}
Thus, for $j\not\in\{i-1, i, i+1\}$, agents $i$ and $j$ are not neighbors. $A_i$ moves to the centroid of $A_i, A_{i-1}, A_{i+1}$ and by symmetry, the locations of the agents are vertices of a new regular $n$-gon centered at $O$ and with a smaller side length $l_{k+1}$. We have

\begin{align*}
l_{k+1} &= \frac{OA_i'}{OA_i}l_k\\
&=\left(1-\frac{2HA_i}{3OA_i}\right)l_k\\
&=\left(1-\frac{2l_k \sin(\frac{\pi}{n})\sin(\frac{2\pi}{n})}{3l_k\sin(\frac{\pi(n-2)}{2n})}\right)l_k\\
&\ge \left(1-\frac{14}{n^2}\right)l_k
\end{align*}

Thus, the system requires at least $n^2/28$ steps to converge.
\end{proof}

\section{Discussion}

In this paper, we analyzed the convergence rate of the homogeneous HK
system in arbitrary dimensions. The system is shown to converge in
polynomial time, but can take at least a quadratic number of steps in
the worst case. Getting tight bounds on the convergence time of the
system, even in just one dimension, remains an interesting open
problem.  

A particularly interesting challenge is to analyze the heterogeneous
version of the HK system, i.e. when the neighborhood radii of all the
agents are not necessarily the same \cite{MB11}. New ideas are needed to understand the
behavior of this system in particular, and directional systems in
general. Our analysis of a noisy variant of homogeneous HK system is a
step towards studying more complicated directional systems. 

Beyond convergence rate, the behavior of the homogeneous HK system is
still full of mysteries. Most notable perhaps is the 2R conjecture \cite{BHT07}:
when agents are drawn uniformly at random on an interval, they
converge to clusters at distance close to twice the minimum possible
inter-cluster distance. Resolving this conjecture remains well beyond
our understanding of the system. 

\bibliographystyle{abbrv}
\bibliography{opinions}

\begin{thebibliography}{10}

\bibitem{Axelrod97}
R.~Axelrod.
\newblock The dissemination of culture: a model with local convergence and
  global polarization.
\newblock {\em Journal of Conflict Resolution}, 1997.
\newblock Reprinted in ``The complexity of cooperation," Princeton University
  Press, Princeton, 1997.

\bibitem{BHT07}
V.~D. Blondel, J.~M. Hendricx, and J.~N. Tsitsiklis.
\newblock On the 2{R} conjecture for multi-agent systems.
\newblock In {\em European Control Conference}, pages 2996--3000, July 2007.

\bibitem{BHT09}
V.~D. Blondel, J.~M. Hendricx, and J.~N. Tsitsiklis.
\newblock On {K}rause's multi-agent consensus model with state-dependent
  connectivity.
\newblock {\em IEEE Transactions on Automatic Control}, 54(11), Nov. 2009.

\bibitem{Chazelle11}
B.~Chazelle.
\newblock The total s-energy of a multiagent system.
\newblock {\em SIAM J. Control and Optimization}, 49(4):1680--1706, 2011.

\bibitem{chazelle-infl}
B.~Chazelle.
\newblock The dynamics of influence systems.
\newblock {\em arXiv:1204.3946v2}, 2012.
\newblock Prelim. version in Proc. 53rd FOCS, 2012.

\bibitem{Ccacm12}
B.~Chazelle.
\newblock Natural algorithms and influence systems.
\newblock {\em CACM Research Highlights}, 2012.

\bibitem{EK10}
D.~A. Easley and J.~M. Kleinberg.
\newblock {\em Networks, Crowds, and Markets - Reasoning About a Highly
  Connected World}.
\newblock Cambridge University Press, 2010.

\bibitem{Fortunato05}
S.~Fortunato.
\newblock On the consensus threshold for the opinion dynamics of
  krause-hegselmann.
\newblock {\em International Journal of Modern Physics C}, 16(2):259--270,
  2005.

\bibitem{HK02}
R.~Hegselmann and U.~Krause.
\newblock Opinion dynamics and bounded confidence: models, analysis and
  simulation.
\newblock {\em J. Artificial Societies and Social Simulation}, 5(3), 2002.

\bibitem{HendrickxB06}
J.~M. Hendrickx and V.~D. Blondel.
\newblock Convergence of different linear and non-linear {V}icsek models.
\newblock In {\em Proc. 17th International Symposium on Mathematical Theory of
  Networks and Systems (MTNS2006)}, pages 1229--1240, July 2006.

\bibitem{Krause00}
U.~Krause.
\newblock A discrete nonlinear and non-autonomous model of consensus formation.
\newblock In {\em Proc. Commun. Difference Equations}, pages 227--236, 2000.

\bibitem{Lorenz05}
J.~Lorenz.
\newblock A stabilization theorem for dynamics of continuous opinions.
\newblock {\em Physica A: Statistical Mechanics and its Applications},
  355(1):217--223, 2005.

\bibitem{Lorenz07}
J.~Lorenz.
\newblock Continuous opinion dynamics under bounded confidence: a survey.
\newblock {\em International Journal of Modern Physics C}, 16(18):1819--1838,
  2007.

\bibitem{MBCF07}
S.~Mart\'{i}nez, F.~Bullo, J.~Cort\'{e}s, and E.~Frazzoli.
\newblock On synchronous robotic networks -- {P}art ii: Time complexity of
  rendezvous and deployment algorithms.
\newblock {\em IEEE Transactions on Automatic Control}, 52(12):2214--2226, Dec.
  2007.

\bibitem{MB11}
A.~Mirtabatabaei and F.~Bullo.
\newblock Opinion dynamics in heterogeneous networks: convergence, conjectures
  and theorems.
\newblock {\em arXiv:1103.2829v2}, Mar. 2011.
\newblock To appear in SIAM J. Control Optim.

\bibitem{Moreau05}
L.~Moreau.
\newblock Stability of multiagent systems with time-dependent communication
  links.
\newblock {\em IEEE Transactions on Automatic Control}, 50(2):169 -- 182, Feb.
  2005.

\bibitem{RMF08}
M.~Roozbehani, A.~Megretski, and E.~Frazzoli.
\newblock Lyapunov analysis of quadratically symmetric neighborhood consensus
  algorithms.
\newblock In {\em CDC}, pages 2252--2257, 2008.

\end{thebibliography}

\end{document}